\documentclass[10pt, conference, compsocconf]{IEEEtran}
\usepackage{cite}
\ifCLASSINFOpdf
 \usepackage[pdftex]{graphicx}
 \graphicspath{ {./experiments/} }
 \DeclareGraphicsExtensions{.pdf,.jpeg,.png}
\else
\fi
\usepackage{subcaption}
\usepackage{multirow}

\usepackage[cmex10]{amsmath}
\usepackage{algorithm, algorithmic}
\usepackage{array}
\usepackage{mdwmath}
\usepackage{mdwtab}
\usepackage{url}
\usepackage{pifont}

\newtheorem{theorem}{Theorem}

\newtheorem{prop}{Proposition}
\begin{document}
%
% paper title
% can use linebreaks \\ within to get better formatting as desired
\title{Dynamic Stale Synchronous Parallel Distributed Training for Deep Learning}

% author names and affiliations
% use a multiple column layout for up to two different
% affiliations

\author{\IEEEauthorblockN{Xing Zhao\IEEEauthorrefmark{1}, Aijun An\IEEEauthorrefmark{1}, Junfeng Liu\IEEEauthorrefmark{2}, Bao Xin Chen\IEEEauthorrefmark{1}}
\IEEEauthorblockA{\IEEEauthorrefmark{1}Department of Electrical Engineering and Computer Science\\
York University,
Toronto, Canada\\
\{xingzhao, aan, baoxchen\}@cse.yorku.ca}
\IEEEauthorblockA{\IEEEauthorrefmark{2}Platform Computing\\
IBM Canada, Markham, Canada\\
jfliu@ca.ibm.com}
}

% conference papers do not typically use \thanks and this command
% is locked out in conference mode. If really needed, such as for
% the acknowledgment of grants, issue a \IEEEoverridecommandlockouts
% after \documentclass

% for over three affiliations, or if they all won't fit within the width
% of the page, use this alternative format:
% 
%\author{\IEEEauthorblockN{Michael Shell\IEEEauthorrefmark{1},
%Homer Simpson\IEEEauthorrefmark{2},
%James Kirk\IEEEauthorrefmark{3}, 
%Montgomery Scott\IEEEauthorrefmark{3} and
%Eldon Tyrell\IEEEauthorrefmark{4}}
%\IEEEauthorblockA{\IEEEauthorrefmark{1}School of Electrical and Computer Engineering\\
%Georgia Institute of Technology,
%Atlanta, Georgia 30332--0250\\ Email: see http://www.michaelshell.org/contact.html}
%\IEEEauthorblockA{\IEEEauthorrefmark{2}Twentieth Century Fox, Springfield, USA\\
%Email: homer@thesimpsons.com}
%\IEEEauthorblockA{\IEEEauthorrefmark{3}Starfleet Academy, San Francisco, California 96678-2391\\
%Telephone: (800) 555--1212, Fax: (888) 555--1212}
%\IEEEauthorblockA{\IEEEauthorrefmark{4}Tyrell Inc., 123 Replicant Street, Los Angeles, California 90210--4321}}

% use for special paper notices
%\IEEEspecialpapernotice{(Invited Paper)}

% make the title area
\maketitle

\begin{abstract}
%The abstract goes here. DO NOT USE SPECIAL CHARACTERS, SYMBOLS, OR MATH IN YOUR TITLE OR ABSTRACT.
Deep learning is a popular machine learning technique and has been applied to many real-world problems, ranging from computer vision to natural language processing. %In most cases deep learning outperformed previous work. 
However, training a deep neural network is very time-consuming, especially on big data. %In addition, the size of a model is growing as neural networks have been designed deeper every few years (e.g. from AlexNet to Inception-v3 and Vgg16).
It has become difficult for a single machine to train a large model over large datasets. A popular solution is to distribute and parallelize the training process across multiple machines using the parameter server framework. In this paper, we present a distributed paradigm on the parameter server framework called \textit{Dynamic Stale Synchronous Parallel} (DSSP) which improves the state-of-the-art \textit{Stale Synchronous Parallel} (SSP) paradigm by 
dynamically determining the staleness threshold at the run time. Conventionally to run distributed training in SSP, the user needs to specify a particular stalenes threshold as a hyper-parameter. However, a user does not usually know how to set the threshold and thus often finds a threshold value through trial and error, which is time-consuming. %SSP is a scheduling scheme for synchronization of distributed and parallel training based on the prevailing parameter server framework used by many distributed platforms. 
Based on workers' recent processing time, our approach DSSP adaptively adjusts the threshold per iteration at running time to reduce the waiting time of faster workers for synchronization of the globally shared parameters (the weights of the model), and consequently increases the frequency of parameters updates (increases iteration throughput), which speedups the convergence rate. % Unlike SSP that uses a fixed user specified threshold for the delay in parameter updates, %DSSP takes a user specified range $R$ and adjusts the staleness threshold $s'$ within the range ($s' \in [s,s+r] : r\in R \subset N$) DSSP dynamically adjusts the threshold to reduce the waiting time of faster workers for the synchronization of the globally shared parameters (i.e., the weights of the model). 
%We designed prediction module in our synchronization controller which
%In the synchronization controller, our prediction module monitors the distributed network traffic from workers to the parameter server and uses their network traffic information to dynamically predict %$r$. an optimal threshold within a given integer range.
We compare DSSP with other paradigms such as \textit{Bulk Synchronous Parallel} (BSP), \textit{Asynchronous Parallel} (ASP), and SSP by running deep neural networks (DNN) models over GPU clusters in both homogeneous and   heterogeneous environments. The results show that 
in a heterogeneous environment where the cluster consists of mixed models of GPUs, DSSP converges to a higher accuracy much earlier than SSP and BSP and 
performs similarly to ASP. 
In a homogeneous distributed cluster, DSSP has more stable and slightly better performance than SSP and ASP, and converges much faster than BSP.

%in the distributed computing network, DSSP converges faster than SSP for ResNet-110 on the CIFAR-100 dataset and it achieves higher accuracy than SSP for a downsized AlexNet on the CIFAR-10 dataset.
%DSSP also saves training time in terms of faster converging to a higher accuracy in certain circumstances such as in a cluster with mixed models of GPUs.
%We present the empirical evaluation of four distributed modes on the GPU cluster and the design of dynamic SSP. A theoretical analysis on regret bound for dynamic SSP is also discussed.
%our method takes less training time than SSP and BSP for large-size DNN models while maintaining a reasonable accuracy close to that of SSP and achieves better accuracies than ASP. 

\end{abstract}


% For peer review papers, you can put extra information on the cover
% page as needed:
% \ifCLASSOPTIONpeerreview
% \begin{center} \bfseries EDICS Category: 3-BBND \end{center}
% \fi
%
% For peerreview papers, this IEEEtran command inserts a page break and
% creates the second title. It will be ignored for other modes.
\IEEEpeerreviewmaketitle

\section{Introduction}
% no \IEEEPARstart
% You must have at least 2 lines in the paragraph with the drop letter
% (should never be an issue)
The parameter server framework{\cite{dean2012large}\cite{chen2015mxnet}} 
%is not a new terminology nowadays. It 
has been developed to support distributed training of large-scale machine learning (ML) models (such as deep neural networks{\cite{krizhevsky2012imagenet}\cite{Szegedy_2016_CVPR}\cite{simonyan2014very}}) over very large data sets, such as Microsoft COCO{\cite{lin2014microsoft}}, ImageNet 1K{\cite{krizhevsky2012imagenet}} and ImageNet 22K{\cite{chilimbi2014project}}.
% and for ML models with a large number of parameters, such as deep neural networks{\cite{krizhevsky2012imagenet}\cite{Szegedy_2016_CVPR}\cite{simonyan2014very}}.
Training a deep model using a large-scale cluster with an efficient distributed paradigm reduces the training time from weeks on a single server to days or hours.

Since the DistBelief{\cite{dean2012large}} framework was developed in 2012, distributed machine learning has attracted the attention of many ML researchers and system engineers.
%to the joint field of machine learning and system engineering. 
In 2014, the Parameter Server architecture {\cite{li2014scaling}} was launched.
%in the OSDI conference. 
Its coarse-grained parallelism shows a significant speedup in convergence over 6000 servers. 
%However, Xing et al.{\cite{xing2016strategies}} emphasized the importance of the fine-grained parallelism for a distributed system and claimed that there was not a general-purpose distributed system for ML algorithms. In order to design an efficient distributed system for an ML algorithm, a careful analysis on the dependency of parameters of the model is required and a customized distributed system can be developed based on groups of dependent and independent parameters. 
In 2015, a fine-grained parallel and distributed system Petuum{\cite{xing2015petuum}} was developed to support customized distributed training for particular machine learning algorithms instead of providing a general distributed framework to many machine learning algorithms based on, e.g., Hadoop{\cite{landset2015survey}} and Spark{\cite{meng2016mllib}}. By then a global competition has begun on developing efficient distributed machine learning training platforms. Baidu published the distributed training system PaddlePaddle{\cite{mao2014deep}} in 2015 for deep learning, which inherits the parameter server framework. Alibaba released KunPeng{\cite{zhou2017kunpeng}} in 2017, an variation of the parameter server, which was claimed as an universal distributed platform. Due to its efficient scalable network communication design, the parameter server framework can be found in most distributed platforms in practice regardless whether they are implemented in fine-grained or coarse-grained parallelism.

In a nutshell, the \textbf{parameter server framework} consists of a logic \textit{server} and many \textit{workers}. Workers are all connected to the server. 
%Thus, the growth of communication cost is $O(n)$ where $n$ is the number of workers. 
The server usually maintains the \textit{globally shared weights} by aggregating weight updates from the workers and updating the global weights. It provides a central storage for the workers to upload their computed updates (by the \texttt{push} operation) and fetch the up-to-date global weights (by the \texttt{pull} operation).  
%The logic server can be implemented in a distributed manner by partitioning a large ML model into smaller ones and distributing them over several servers for load balancing and relieving the bottleneck on network communication. 
The parameter server framework supports \textit{model parallelism} and \textit{data parallelism} {\cite{lee2014model}\cite{zhou2016convergence}}. In model parallelism, a ML model can be partitioned and its components are assigned to a server group (i.e., a distributed logic server) and workers. A worker computes the gradients for a server based on its assigned model partition and data. However, it is difficult to decouple a model due to dependencies between components of the model (e.g., layers of DNN) and the nature of the optimization method (e.g., stochastic gradient descent){\cite{xing2016strategies}}. Thus, model parallelism is not commonly seen in practice. In this work, we focus on data parallelism, in which the training data is partitioned based on the number of workers and a partition is assigned to each worker. A worker machine usually contains a replica of the ML model and is assigned a equal-sized partition of the entire training data. Each worker iterates the following steps: \ding{172} computing the gradients based on a sample or a mini-batch and its local parameters (e.g., weights), \ding{173} sending the gradients as an update to the server, \ding{174} retrieving the latest global weights from the server and \ding{175} assigning the retrieved weights as its local weights. 
%Therefore, we call an \textit{iteration interval} of a worker is 
We call the time span accumulated from \ding{172} to \ding{175} the \textit{iteration interval}. From the server's perspective, an iteration interval of a worker is the time period between 
%the previous update and the current update 
two consecutive updates it receives from the worker. 
%This concept is used in our method described later in the paper.

%The parameter server framework supports \textit{model parallelism} and \textit{data parallelism} {\cite{lee2014model}\cite{zhou2016convergence}}. For model parallelism, an ML model can be partitioned and its components are assigned to a server group and workers. Workers will compute the gradient for each server based on its assigned model partition and data. However, it is difficult to decouple a model due to dependencies between components of the model (e.g. layers of DNN) and the nature of the optimization method (e.g., stochastic gradient descent){\cite{xing2016strategies}}. Thus, model parallelism is not commonly seen in practice. In data parallelism, the training data is partitioned based on the number of workers and a partition is assigned to each worker. Each worker has a replica of the training model. The parameter server maintains the globally shared up-to-date weights (of the model) for all workers. The workers compute the gradients based on its assigned data, send the gradients to the server, and then pull the updated weights from the server at the end of each iteration. In practice, data parallelism is the most practical paradigm for the parameter server framework in terms of scaling up the number of workers in a distributed network{\cite{zhang2017comparison}}. Hence, data parallelism is commonly used in the parameter server framework for distributed deep neural network training{\cite{chung2017parallel}\cite{wang2018bigdl}}.

\subsection{Distributed paradigms for updating the parameters}
There are three paradigms for updating the model parameters during distributed deep learning with data parallelism. 
%ost of the aforementioned distributed systems for ML provide three paradigms of updating the parameters during training. 
They are Bulk Synchronous Parallel (BSP){\cite{gerbessiotis1994direct}}, Stale Synchronous Parallel (SSP){\cite{ho2013more}\cite{cui2014exploiting}} and Asynchronous Parallel (ASP){\cite{dean2012large}\cite{recht2011hogwild}\cite{chilimbi2014project}}. 
%BSP-like paradigms generally exist in many parallel and concurrent computing programs for managing the threads. Each thread computes few lines of code in parallel but has to wait for other threads at a join point for the dependency of the code in context. This join point can be considered as a synchronization point of threads. The idea has been successfully implemented in the distributed system. Meanwhile, the waiting time is also inherited. A strangler (the slowest machine) makes the others wait for it, which is noxious. The goal of SSP and ASP is to increase the computing time by allowing more weight updates (thus increasing iteration throughput) and decrease the waiting time for communication for the parameter server framework{\cite{dai2015high}}. The waiting time for communication can be further divided into the waiting time for synchronization of workers and the data (i.e., parameters) transmission time between workers and the server. We will cover each paradigm in detail and propose our framework in the paper.

\subsubsection{Bulk Synchronous Parallel (\textbf{BSP})}
In BSP, all workers send their computed gradients to the parameter server for the global weight update and wait for each other at the end of every iteration for synchronization. 
%Thus, at the end of each iteration, the %global 
%weights are the same for all workers. 
Once the parameter server receives gradients from all workers and updates the global weights, it sends the latest global weights (parameters) to the workers before each worker starts a new iteration. In this paradigm, every worker starts a new iteration based on the same version of the global weights from the server, that is, the weights are {\it consistent} among all workers. 
%From the system's perspective, the global weights are consistent among all workers at the beginning of their corresponding iterations. 
%
BSP generally achieves the best accuracy among the three paradigms but takes the most training time due to its frequent synchronization among workers. Synchronization incurs waiting time for faster workers. 
%More times of synchronization bring more waiting time to faster workers. 
%SSP came after BSP and ASP. It’s a combination balanced between BSP and ASP. SSP has better accuracy than ASP but less than BSP. The training time of SSP is more than ASP but can be a lot shorter than BSP in a large size of a cluster.
%
%BSP is the dominant distributed model in applications of many disciplines other than ML{\cite{malewicz2010pregel}\cite{siddique2016apache}}. Its most profound application in the distributed framework field is MapReduce{\cite{senger2016bsp}\cite{pace2012bsp}}. Most general propose frameworks such as Hadoop{\cite{landset2015survey}} and Spark{\cite{meng2016mllib}} all use MapReduce-like parallel training. That is, all machines (also called workers in distributed training framework) have to wait for each other to finish the computing and updates by the end of each iteration for synchronization. A signification amount of time is wasted in such model. 
Since ML has the fault tolerant property{\cite{xing2016strategies}} (that is, it is robust against minor errors in intermediate calculations) when it uses the iterative convergent optimization algorithm such as stochastic gradient descent (SGD){\cite{zinkevich2010parallelized}},
%which supports parallel training, 
a more flexible paradigm that uses less synchronization can be applied.

\subsubsection{Asynchronous Parallel (\textbf{ASP})}~\label{asp_converage}
In ASP, all workers send their computed gradients to the parameter server at each iteration but no synchronization is required. Workers do not wait for each other and simply run independently during the entire training. In this case, some slower workers will bring the delayed or staled gradient updates to the globally shared weights on the parameter server. The delayed gradients introduce errors to the iterative convergent method. Consequently, it prolongs the convergent rate of a training model and even diverges the learning of the model when the staled updates are from very old iterations.
Without any synchronization, each worker may obtain a different version of the global weight from the parameter server at the end of its iteration. From the system's point of view, the global weights are {\it inconsistent} to all workers at the beginning of their iterations.
In contrast to BSP, ASP has the least training time for a given number of epochs, but usually yields a much lower accuracy due to the missing synchronization step among workers. Nonetheless, zero or less synchronization for workers usually diffuses the convergence of 
%the gradient descent method of 
a DNN model{\cite{wang2017accelerating}}. Hence, ASP is not stable in terms of the model convergence.

\subsubsection{Stale Synchronous Parallel (\textbf{SSP})}
SSP is a combination of BSP and ASP. It relies on a policy to switch between ASP and BSP dynamically during the training. The policy is to restrict the number of iterations between the fastest worker(s) and the slowest worker(s) to not exceed more than a user-specified staleness threshold $s$ where $s$ is a natural number. This policy ensures that the difference in the number of iterations among workers is no larger than $s$. Hence, as long as the policy is not violated, 
%(i.e., the threshold among workers is not exceeded), 
%no synchronization is required among the workers. Therefore, 
there is no waiting time among workers. When the threshold is exceeded, the synchronization is forced on fastest workers which are $s$ iterations ahead of the slowest worker. %Just like BSP, faster workers have to wait for the slowest one by the end of their iterations. 
One effective implementation of SSP from {\cite{ho2013more}} only forces the fastest worker(s) to wait for the slowest worker(s) and allows the rest continue their iterations. In this distribution paradigm, the parameters (the global weights) of an ML model are considered \textit{inconsistent} \cite{dai2015high} among the workers at the beginning of their iterations when the policy is not violated. However, the inconsistency is limited to a certain extent by the \textit{small} threshold $s$ so that the ML model can still converge (close to an optimum){\cite{wei2015managed}}.

\subsection{Contributions}

SSP is an intermediate solution between BSP and ASP. It is faster than BSP, and guarantees convergence, leading to a more accurate model than ASP. However, in SSP the user-specified staleness threshold is fixed, which leads to two problems. First, it is usually hard for the user to specify a good single threshold since user has no knowledge which value is the best. Choosing a good threshold may involve manually searching in an integer range via numerous trials. Also, a DNN model involves many other hyperparameters (such as the number of layers and the number of nodes in each layer). When these parameters change, the same searching trials have to be repeated again to fine-tuning the staleness threshold. Second, a single fixed value may not be suitable for the whole training process. An ill-specified value may cause the fastest workers to wait for longer time than necessary.
%Neither SSP nor bounded delay takes fine-grained scheduling operation but only balances \textit{computing} intervals and \textit{waiting for communication} intervals of workers to achieve higher GPU utilization. Nonetheless, there is still room in the {\it waiting for communication} interval that can be further reduced when a boundary is forced on the fastest worker. 
For example, Figure \ref{fig:worker_intervals} shows that if the threshold is exceeded at the red solid line, the waiting time for the fastest worker if it starts waiting right away is more than the waiting time if it continues but starts waiting at the green solid line. In fact, the waiting time for it to start waiting at the yellow solid line is the minimum of the three. However, we have to make sure that the difference in iterations between the fastest and slowest workers is not too large. Otherwise, too many staled updates may delay the convergence of the ML model{\cite{li2014communication}}. 

To solve these problems, we propose an adaptive SSP scheme named \textit{Dynamic Stale Synchronous Parallel} (\textbf{DSSP}). Our approach dynamically selects a threshold value from a given range in the training process based on the statistics of the real-time processing speed of distributed computing resources. It allows the threshold to change over time and also allows different workers to have different threshold values, adapting to the run-time environment.
%In DSSP, we use a range parameter $R$ to allow flexibility in the degree of staleness and 
To achieve this purpose, we design a \textit{synchronization controller} at the server side to 
%bound the fastest worker(s) at the end of certain iteration to wait for the slowest worker(s). The synchronization controller is built into the parameter server. It 
determine how many iterations the current fastest worker should continue running at the end of its iteration upon the excess of the lower bound of a user-specified staleness threshold range. 
%Thus, the iteration throughput is increased. Consequently the convergence rate is increased. (There is no theory supporting this "consequence")
The decision is made at run time by estimating the future waiting times of workers based on the timestamps of their previous \textit{push requests} and selecting a time point in the range that leads to the least estimated waiting time. In this way, we enable the parameter server to dynamically adjust or relax the synchronization of workers during the training based on the run-time environment of the distributed system.

In addition, although experiments have been reported on parameter servers with a variety of ML models, experiments of comparing DNN models under different distributed paradigms are rarely seen in the literature. In this paper, we look into four distributed paradigms (i.e., BSP, ASP, SSP and DSSP) and compare their performance by training three DNN models on two image datasets using MXNet{\cite{chen2015mxnet}} which provides BSP and ASP. 
%Our contribution is to develop an adaptive synchronization SSP scheme named \textit{Dynamic Stale Synchronous Parallel} (\textbf{DSSP}). We compare DSSP with SSP and two other paradigms (BSP and ASP). 
We implemented both SSP and DSSP in MXNet, report and analyze our findings from the experiments. 
%We implement SSP following {\cite{ho2013more}\cite{li2014communication}} based on MXNet and further develop DSSP on MXNet. 
%BSP and ASP are already available on MXNet.

\section{Related Work}
% more works, describe them only
There are variety of approaches to optimizing the distributed paradigms under the parameter server framework. Generally, they can be categorized into three basis streams: BSP, ASP and SSP. Chen et al. {\cite{chen2016revisiting}} try to optimize the BSP by adding few backup workers. That is to train $N$ workers in BSP, they add $c$ backup workers so that there are $N + c$ workers during the training. By the end of each iteration, the server only takes the first $N$ arrived updates and drops the $c$ slower arrived updates from the stranglers for weight synchronization. In this case, the training data allocated to the $c$ random slower workers are partially wasted in each iteration.

%optimise ASP
Hadjis et al. {\cite{hadjis2016omnivore}} optimize ASP from machine learning's perspective. It adjusts the momentum based on the degree of asynchrony (staleness of the gradients). Then, it uses the tuned momentum to mitigate the divergent direction that staled gradients introduce.  
Meanwhile, model parallel computing is applied here for better performance where a DNN model is split into two parts: convolutional layers and fully connected layers. Both parts are computed parallelly and concurrently.

Zhang and Kwok{\cite{zhang2014asynchronous}} propose to use asynchronous distribution to optimize synchronous ADMM algorithm. However, it uses \textit{partial barrier} to make the fastest workers wait for the slowest workers and \textit{bounded delay} to guarantee that the iterations among workers do not exceed a user specified hyperparameter $\tau$ which is equivalent to the staleness threshold $s$ of SSP in {\cite{ho2013more}}. Then, all these make the approach rather close to SSP than ASP optimization. Bounded delay also appears in {\cite{li2014communication}} and is elaborated in the rest of this section.

%Eric Xing's SSP
%Vanilla SSP was proposed in {\cite{ho2013more}}., which requires that the difference in iterations between the fastest and slowest workers cannot exceed a user specified staleness threshold $s$. The synchronization operation for the fastest worker is determined by the threshold $s$ at the run time. %For example, at $t$th iteration, the fastest worker has to wait for workers running behind by $t-s-1$ iterations. 
%This condition serves for the synchronization purpose. 
%The threshold limits the error (staled) update caused by latency. 
%%Here $t$ can be considered as the $t$th iteration a worker is running at or the number of iterations a worker has completed (in this case, $t-s$ should be used in the earlier context). 
%In this approach, whether the specified staleness threshold $s$ is exceeded is measured only based on the fastest worker and the slowest worker. %In other words, only the iterations of the fastest and slowest workers matter to the evaluation of the 
%synchronization condition and need to be memorized by the system. While the fastest worker needs to wait
%%is bounded or stopped continuing a new iteration 
%upon the excess of the threshold, the other workers are not affected and can continue their new iterations. We take advantage of this strategy in our work but use a more relaxed scheduling strategy for synchronization on the fastest worker to achieve more iteration throughput by reducing its waiting time.

Li et al.{\cite{li2014communication}} introduce \textit{bounded delay} which is similar to SSP except that it takes all workers' iterations into account instead of letting each worker count its own iteration. In order to keep the ML model (global weights) consistent, iterations in sequence are allowed to run concurrently in parallel under the dependency restriction. 
%The dependency describes the relationship between two entities is either dependent or independent. 
Iteration $t$ depends on iteration $t-k$ if iteration $t$ requires the result of iteration $t-k$ in order to proceed. %is absolute SSP which forces all workers to wait for synchronization when the staleness threshold is exceeded. 
In the bounded delay approach, the number of bounded iterations $k$ is specified by the user, similar to the staleness threshold $s$ in SSP. $k$ means that for a continuous $k$ iterations, every iteration is independent of each other, and they can run concurrently in parallel without waiting for each other. 
%Such a state is equivalent to ASP.
When $k$ is exceeded, the fastest iteration $t$ has to wait for the slower iterations behind $t-k$. Imagine iterations are pre-assigned to workers as tasks, then the bounded delay is equivalent to SSP. For example, we have iterations $\{I_1,I_2,I_3,I_4,I_5,I_6\}$ and two workers $P_1$ and $P_2$. Each iteration $I_i$ completes a min-batch of samples. $P_1$ receives $\{I_1,I_3,I_5\}$, $P_2$ receives $\{I_2,I_4,I_6\}$. If $k=3$, then $I_4$ depends on $I_1$, $I_5$ depends on $I_2$, $I_6$ depends on $I_3$. So $P_2$ at $I_4$ has to wait for $P_1$ finishes $I_1$. $P_1$ at $I_5$ has to wait for $P_2$ completes $I_2$. Bounded delay is rather an inflexible scheme by pre-scheduling tasks to workers. Although the authors briefly claim that more consistent paradigms can be developed based on the dependency constraint, no further exploration is provided in their paper. Our work extends this direction and presents a flexible scheduling approach in which $k$ is dynamically assigned at the training time. In our dynamic bounded delay paradigm, every optimal bound $k$ yields the least waiting time for workers by optimizing the synchronization frequency. For the continuous $k$ iterations in which every iteration is independent is adjusted dynamically in the running time to reduce waiting time of coming iterations which depend on the earlier ones.
%The authors in {\cite{li2014communication}} briefly mention that there are more advanced consistency models can be explored based on the tasks dependency. But no further exploration regard to this direction is provided.Our work offers an extension towards the direction and develops an advanced model. DSSP can also be considered as a dynamic bounded delay which allows dynamically searching the optimal bound $k$ within a given range. That is, for $k-$bounded delay, $k$ is dynamically determined during the model training time based on the iteration time-stamps of the workers. The optimal bound $k$ yields the least waiting time for workers by reducing the synchronization frequency.
%For example, we define a lower threshold $k_B$ and a range $R$, then $k = k_B + r$ where $r \in R$. By collecting the iteration times of each worker, we allow the fastest worker run $r$ more iterations instead of waiting for the slowest worker at $kth$ iteration. $r$ is determined by the most recent iteration times sampled from both the fastest worker and the slowest worker. It searches the time point so that by the end of added $k’$ iterations of the fastest worker, the slowest worker is also approaching the end (not the middle or the beginning) of its running iteration. Thus, it minimizes the length of the waiting time of the fastest worker. Our approach will further save the training time comparing to the $k$-bounded delay when the number of iterations is large.

%\iffalse

\begin{figure}[!ht]
    \centering
    \includegraphics[ width=0.483\textwidth]{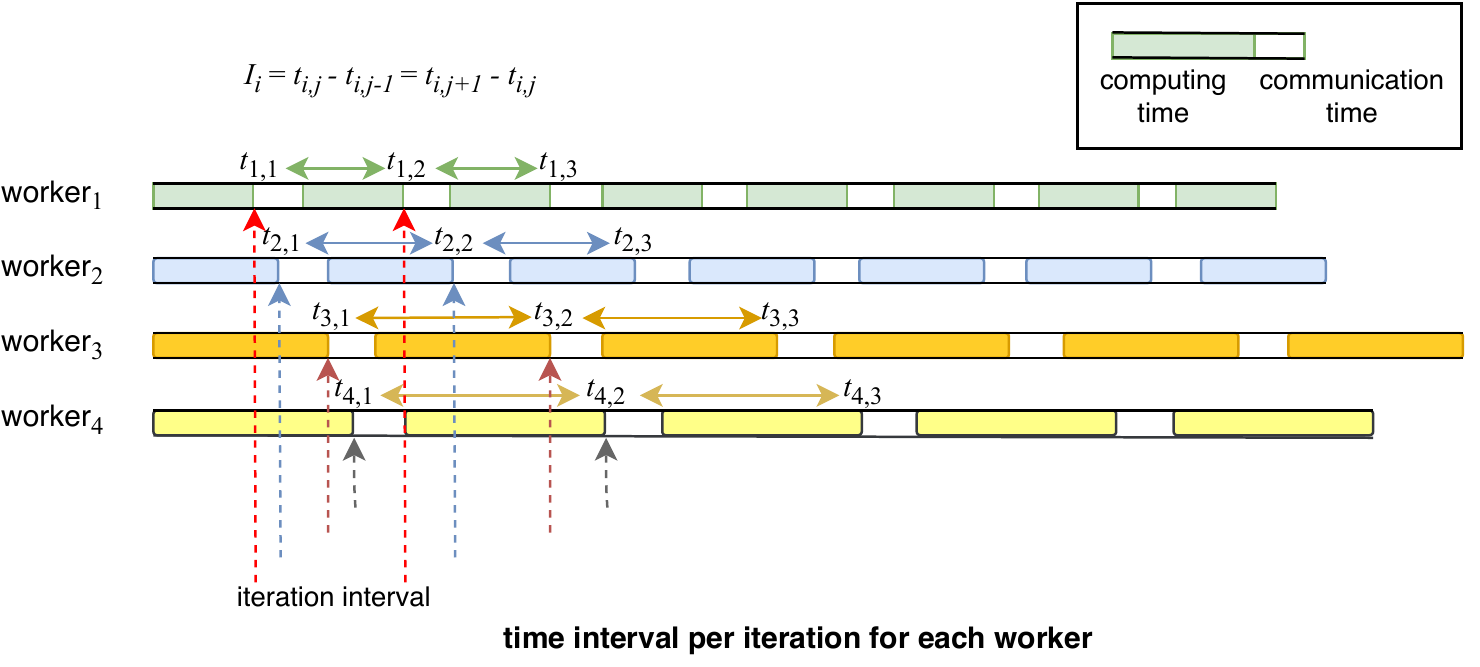}
    \caption{Iteration \textit{interval}s measured by timestamps of \texttt{push} requests from workers. A dotted line represents the time for a \texttt{push} request from a worker to the server. An \textit{interval} consists of \textit{communication period} (blank block) and gradient \textit{computation period} (solid block). %Waiting for communication is composed by data transmission time and waiting for synchronization time.
    }
    \label{fig:push_intervals}
\end{figure}
\begin{figure*}[!ht]
    \centering
    \includegraphics[width=0.88\textwidth]{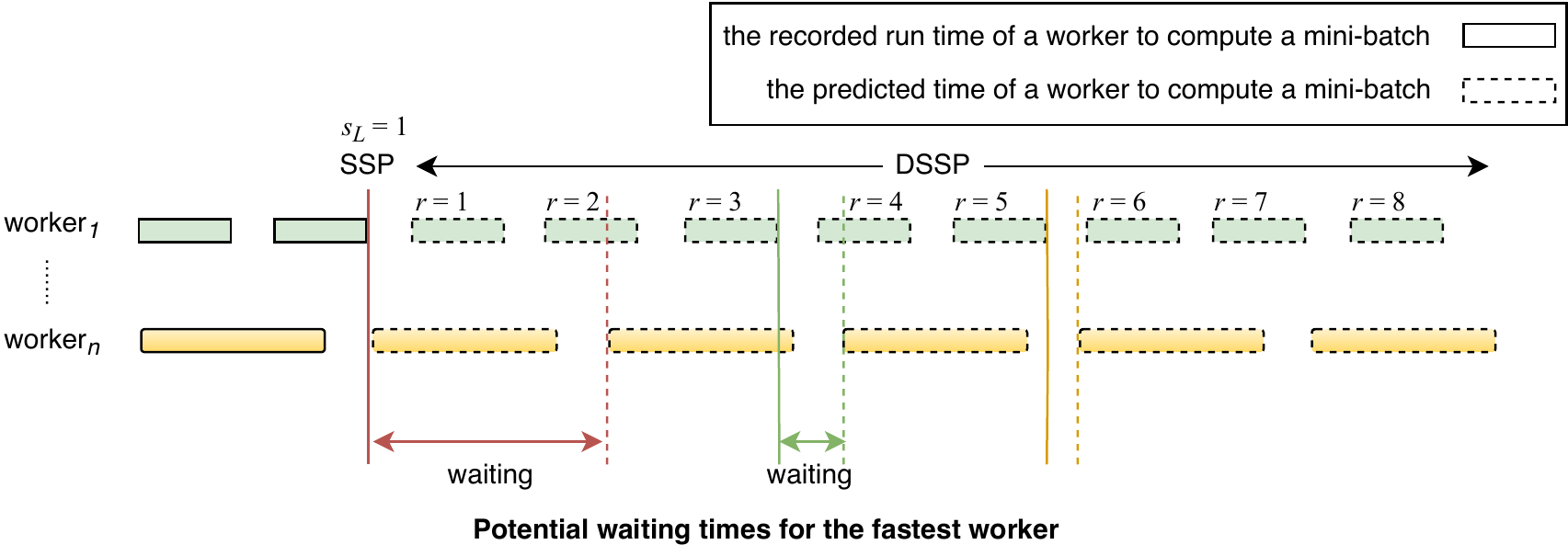}%width=0.483 for single col
    \caption{Prediction module finding the \textit{least waiting time} for the \textit{fastest worker} via iteration time intervals of workers. A solid line represents a boundary to stop the fastest workers continuing new iterations for \textit{synchronization} and a dash line represents the end of waiting when the slowest worker completes its running iteration. The solid line is drawn upon a fastest worker sends a \texttt{push} request to the server and waits for the \texttt{\em OK} signal from the server. Once \texttt{\em OK} is received, it pulls the new updated weight from the server and starts a new iteration where the dash line is drawn. The dash line also indicates the time that the \textit{slowest worker} receives a new updated weight via \texttt{pull} request and starts a new iteration. Worker$_1$ is the \textit{fastest worker} and the worker$_n$ is the \textit{slowest worker}. The colored block represents one iteration time. Following SSP, worker$_1$ has to stop at the red solid line. DSSP compares each $r$ value and finds the $r^*$ which gives the least waiting time. Here, $r^*=3$ if $r \in R=[0,4]$. DSSP allows worker$_1$ to run $3$ more iterations and stop at the green solid line.}
    \label{fig:worker_intervals}
\end{figure*}

%\fi
\iffalse
\begin{figure*}
 \begin{subfigure}{\textwidth}
  \centering
  \includegraphics[width=0.85\linewidth]{workers_iterations.pdf}
  \caption{Finding the least waiting time for the fastest worker via iteration time cost intervals of workers}
  \label{fig:worker_intervals}
\end{subfigure}
 \medskip
 \begin{subfigure}{\textwidth}
  \centering
  \includegraphics[width=0.85\linewidth]{workers_iteration_push_requests.pdf}
  \caption{Iteration time interval measured by time-stamps of push request from workers}
  \label{fig:push_intervals}
\end{subfigure}
\caption{Predict optimal iteration for synchronization}
\label{fig:iterations}
\end{figure*}
\fi

\section{Dynamic Stale Synchronous Parallel} \label{DSSP_approch}
%\subsection{SSP analysis}
% move to approach
% need figure to explain
% Based on heavy traffic approximation from queue theory, the recent history has more representation for an ongoing event. Thus, a straight approach is to dynamically memorize the most recent period of history of each processing time interval of an iteration per GPU. Then, we can predict the best spot based on the collected time intervals per iteration per GPU.
In this section, we propose the \textit{Dynamic Stale Synchronous Parallel} (\textbf{DSSP}) synchronization method. Instead of using a single and definite staleness threshold as in SSP, DSSP takes a range for the staleness threshold as input, and dynamically determines an optimal value for the staleness threshold during the run time. The value for the threshold can change over time and adapt to the run-time environment. 

%We obtain the information about workers' iterations from the server and predict the optimal spot on the time line to bound the fastest worker for synchronization.
%We develop a prediction method to simulate the imminent intervals of iterations per worker.
%We propose a method for predicting %the imminent iteration interval for each worker and use it to predict 
%the optimal spot on the time line to bound the fastest worker for synchronization. 
\subsection{Problem statement}
%\begin{prop}
Given a \textit{lower bound} and an \textit{upper bound} of staleness thresholds $s_L$ and $s_U$, %and a range $R=[0,r_{max}]$ of natural numbers, 
DSSP finds an \textit{optimal} threshold $s^* \in [s_L, s_U]$ for a worker dynamically,
%$s^* \in [s, s+r_{max}]$ 
which yields the \textit{minimum waiting time} for the worker to \textit{synchronize} with others, based on the \textit{iteration time} collected from each worker at the run time. 

In other words, DSSP finds an integer $r^* \in R$, where $R=[0, r_{max}]$,  $r_{max}=s_U - s_L$ and $r^*=s^*-s_L$. That is, DSSP finds an \textit{optimal spot} on the time line $R$ to bound the workers for synchronization. Empirically we can find a $r=s-s_L$ that is close to $r^*$, which is the same as finding $s \in [s_L, s_U]$ close to $s^*$.

%As the waiting time is minimized, workers spend more time on computing iterations. Thus, the iteration throughput is increased which consequently increases the convergence rate.   (No theoretical support for this claim or causal relationship)
%we can find a threshold $s' = s_L + r, r \in R$ which is close to $s^*$. We know $s' \in [s_L, s_U]$ by definition.
%\end{prop}

\subsection{Assumption} ~\label{assumption}
%Our prediction of workers' iteration intervals is based on the assumption 
An \textit{iteration interval} of a worker is the time period between 
two \textit{consecutive} updates (i.e., \texttt{push} requests) the server receives from the worker. 
We can measure the length of an iteration interval of a worker by using the timestamps of the \texttt{push} requests sent by the worker (see Figure \ref{fig:push_intervals}). 

We assume that the iteration intervals of a worker in continuous iterations in a short time period are very similar. That is, for contiguous iterations of a worker in a short time period, each iteration has the similar processing time which includes computing gradients over a mini-batch, sending gradients to and receiving updated weights (parameters) from the parameter server. 
Therefore, if we use \textit{the most recent intervals} to estimate the length of next intervals, the error is small under this assumption. Even if the network experiences some instabilities in a short period and we may make some wrong predictions which may lead to some latent updates, %DSSP will continue to predict correctly once the network becomes stable since it always uses the most recent observed intervals to make predictions. 
DSSP can still converge due to the error tolerance of an iterative-convergent method such as Parallelized SGD{\cite{zinkevich2010parallelized}}. 

\subsection{Method}
\iffalse
\begin{prop}Given a lower bound and an upper bound of staleness thresholds $s_L$ and $s_U$. %and a range $R=[0,r_{max}]$ of natural numbers, 
DSSP finds an optimal threshold $s^* \in [s_L, s_U]$
%$s^* \in [s, s+r_{max}]$ 
which yields the minimum waiting time of the fastest workers for synchronization based on the iteration time collected from each worker at running time. Suppose we take $s_L$ as the lower threshold, then we can find a integer $r \in [0, r_{max}]=R$ where $r_{max}=s_U - s_L$ and $r_{max}$ is the largest value in $R$. Empirically we can find a threshold $s' = s_L + r, r \in R$ which is close to $s^*$. We know $s' \in [s_L, s_U]]$ by definition.
\end{prop}
\fi

%\iffalse
\begin{algorithm}
 \caption{Dynamic Staled Gradient Method}
 \label{alg:dssp_ps}
 \begin{algorithmic}[1]
 \item[]\hspace{-0.5cm}\textbf{\underline{Worker $p$ at iteration $t_p$}}
 %\item[] \textbf{Initialization:} $r_p \leftarrow 0$
% \IF{($t_p=1$)}
%  \STATE $r_p=0$
% \ENDIF
 %\WHILE{receiving $OK$ from Server}
   \STATE Wait until receiving \texttt{\em OK} from Server
   \STATE \texttt{pull} weights $w^s$ from Sever
   \STATE Replace local weights $w^{t_p}$ with $w^s$ 
   \STATE Gradient $g^{t_p} \leftarrow \frac{1}{m}\sum_{i=1}^{m} \partial l_{loss}((x_i, y_i), w^{t_p})$ \\\COMMENT{$m$: size of mini-batch $M$ and $(x_i,y_i) \in M$}
  \STATE \texttt{push} $g^{t_p}$ to Server
  %\STATE{Wait for $OK$ from Server}  %\COMMENT{ACK, acknowledge}
 %\ENDWHILE
 \end{algorithmic} 
 
 \begin{algorithmic}[1]
 \item[]\hspace{-0.5cm}\textbf{\underline{Server at iteration $t_s$}}
 \item[-] Upon receiving \texttt{push} request with $g^{t_p}$ from worker $p$;
 \item[-] $r_p$ stores the number of extra iterations worker $p$ is allowed beyond $s_L$, initialized to zero at the very beginning;
 \item[-] $t_i$ stores the number of \texttt{push} requests received from worker $i$ so far
 %\IF{($r_p = 0$)}
 \STATE $t_p=t_p+1$
 \STATE Update the server weights $w^{t_s}$ with $g^{t_p}$. If some other workers send their updates at the same time, their gradients are \textit{aggregated} before updating $w^{t_s}$
 %$w^{t_s} \leftarrow w^{t_s-1} - \eta (g^{t_p} + \nabla f(w^{t_s-1}))$
 \IF{($r_p>0$)}
   \STATE $r_p=r_p-1$
   \STATE Send \texttt{\em OK} to worker $p$ 
 \ELSE
   \STATE Find the {\em slowest} and {\em fastest} workers based on array $t$
   %\STATE Find the slowest worker based on array $t$
   \IF{($t_p-t_{slowest}\le s_L$)}
     \STATE Send \texttt{\em OK} to worker $p$
   \ELSE
     \IF{$t_p$ is the {\em fastest} worker}
       \STATE $r_p \leftarrow$ synchronization\_controller ($clock_p^{push}, r_p$) \\ \COMMENT{$clock_p^{push}$: timestamp of \texttt{push} request from worker $p$}
       \IF{($r_p>0$)}
        \STATE Send \texttt{\em OK} to worker $p$
       \ENDIF
    \ENDIF
    %\ELSE 
       \STATE Wait until the {\em slowest} worker sends the next \texttt{push} request(s) so that $t_p-t_{slowest}\le s_L$. After updating the server weights $w^{t_s}$ with (\textit{aggregated}) gradients, send \texttt{\em OK} to worker $p$
     %\ENDIF
   \ENDIF
% \ELSE
%   \STATE $r_p=r_p-1$ \COMMENT{when $r_p>0$}
%   \STATE Send $OK$ to worker $p$ 
 \ENDIF
 \end{algorithmic} 
\end{algorithm}

% end block

\begin{algorithm}
 \caption{synchronization\_controller}
 \label{alg:sync_contrl}
 \begin{algorithmic}[1]
  \renewcommand{\algorithmicrequire}{\textbf{Input:}}
 \renewcommand{\algorithmicensure}{\textbf{Output:}}
 \REQUIRE $push_p^t$: timestamp of \texttt{push} request of worker $p$ for sending its iteration $t$'s update to the server\\
 %$r_p$: number of extra iterations that worker $p$ is scheduled to run
 \ENSURE $r^*$, number of extra iterations that worker $p$ is allowed to run
 \item[] \{Table $\mathcal{A}$ stores the timestamps of two latest \texttt{push} requests by all workers, where $\mathcal{A}[i][0]$ stores the timestamp of the latest \texttt{push} request by worker $i$ and $\mathcal{A}[i][1]$ stores the timestamp of the second latest \texttt{push} request by worker $i$\}
  \STATE $\mathcal{A}[p][1] \leftarrow \mathcal{A}[p][0]$
  \STATE $\mathcal{A}[p][0] \leftarrow push_p^t$ 
  \STATE Find the {\em slowest worker} from table $\mathcal{A}$
  \STATE Compute the length of the latest iteration interval of worker $p$: $\mathcal{I}_p \leftarrow \mathcal{A}[p][0] - \mathcal{A}[p][1]$
  \STATE Compute the length of the latest iteration interval of the {\em slowest worker}:  $\mathcal{I}_{slowest} \leftarrow \mathcal{A}[slowest][0] - \mathcal{A}[slowest][1]$
  \STATE  Simulate next $r_{max}$ iterations for worker $p$ based on $\mathcal{I}_p$ and $\mathcal{A}[p][0]$ by storing the $r_{max}$ simulated timestamps in array $Sim_p$ so that: \\
  $Sim_p[0] \leftarrow \mathcal{A}[p][0]$ \\
  $Sim_p[i] \leftarrow Sim_p[0] + i \times \mathcal{I}_p$ where $ 0 < i \le r_{max}$
  \label{sim_table_start}
  \item[] \COMMENT{$r_{max}$: the maximum extra iterations allowed} \label{simulate_size}
  \STATE Repeat the above step for the {\em slowest worker} and store the $r_{max}$ simulated timestamps in array $Sim_{slowest}$ with $Sim_{slowest}[0] \leftarrow \mathcal{A}[slowest][0] + \mathcal{I}_{slowest}$
  
  %\FOR{$i=1$ to $r_{max}+1$}
  %  \FOR{$j=0$ to $r_{max}$}
  %      \STATE $d[i,j]=Sim_{slowest}[i]-Sim_p[j]$
  %  \ENDFOR
  %   \STATE Find $j$ whose $d[i,j]$ is the smallest
  %   \STATE $D[i]=d[i,j]$
  %   \STATE $k[i]=j$
  %\ENDFOR
  %\STATE Find $i$ whose $D[i]$ is the smallest
  %\STATE $r^*=k[i]]$
  \STATE Find the simulated time point $r^*$ for the index of $Sim_p[r]$ that minimizes $|Sim_{slowest}[k]-Sim_p[r]|$ for all $k\in[0,r_{max}]$ and $r\in[0,r_{max}]$
  \RETURN $r^*$ 
  \end{algorithmic} 
\end{algorithm}
%\fi

The proposed DSSP method is described in Algorithm \ref{alg:dssp_ps}.
The algorithm contains two parts: one for workers and the other for the server. Each worker is assigned a partition of the training data, and computes parameter updates (i.e., gradients) iteratively with the partition using the stochastic gradient descent (SGD) method. In each iteration, a mini-batch of the partition is used to compute the gradients based on the current local weights. The worker then sends the gradients to the server through a {\em push request} and waits for the server to send back the \texttt{\em OK} signal. After receiving \texttt{\em OK}, the worker pulls the weights from the server and replaces its local weights with the global weights from the server. The training at the worker continues with the next mini-batch of the data partition based on the new weights.

On the server side, once the server receives a \textit{push request} from a worker $p$, it updates its weights with the gradients from worker $p$. It then determines whether to allow worker $p$ to continue. If yes, it will send worker $p$ an \texttt{\em OK} signal; otherwise, it postpones sending the \texttt{\em OK} signal until the slowest worker catches up.

To determine whether to allow a worker $p$ to continue, the server stores \textit{the number} of \textit{push requests} received from each worker and finds the \textit{slowest worker}. If the number of push requests of worker $p$ is no more than $s_L$ iterations away from the slowest worker, the server allows worker $p$ to continue by sending \texttt{\em OK} to worker $p$ (Lines 7-9 in Algorithm 1). Otherwise, if worker $p$ is currently the \textit{fastest worker}\footnote{The reason we call the procedure only for the current fastest worker is to save the server's computation time.}, the server calls the {\em synchronization\_controller} procedure to determine whether it allows worker $p$ to continue with extra iterations. 

Algorithm 2 describes the {\em synchronization\_controller} procedure. It stores in table $\mathcal{A}$ the \textit{timestamps} of the \textit{two} latest \textit{push requests} from all workers, and uses the information in $\mathcal{A}$ to simulate the next $r_{max}$ iterations of worker $p$ and the slowest worker, where $r_{max}$ is the maximum number of extra iterations allowed for a worker to be ahead of the slowest worker beyond the lower bound of the staleness threshold. With the simulated timestamps, it finds a time point $r^*$ in the range of $[0, r_{max}]$ that minimizes the simulated waiting time of worker $p$ (Line 8 in Algorithm 2). The value $r^*$ is returned to the caller (the Server part of Algorithm 1) and stored as $r_p$ for worker $p$. For example, in Figure \ref{fig:worker_intervals}, suppose worker $n$ is the slowest worker and we are currently processing worker $1$ (i.e., $p=1$). The green boundary yields the least waiting time for worker $1$. Then worker $1$ should continue running 3 more iterations once $s_L$ is exceeded. In this case, 3 is the $r^*$ returned to the server procedure of Algorithm 1. If 0 is returned, it indicates that the current iteration yields the least waiting time, and the worker should wait for the slowest worker at the current iteration. 

In future iterations when worker $p$ sends a push request, if $r_p>0$, the server sends the \texttt{\em OK} signal right after updating the global weights with the gradients sent by the worker and decreases $r_p$ by $1$. In this way, even if worker $p$ is not the fastest worker in that iteration of the server, as long as its $r_p>0$ (due to being the fastest worker in a previous iteration), it can still perform extra iterations beyond $s_L$. Thus, our method is flexible in that different workers may have different thresholds, and also the threshold for a worker can change over time, depending on the run-time environment.

%block ends

\section{Theoretical Analysis}
In this section, we prove the convergence of SGD under DSSP by showing that DSSP shares the same \textit{regret bound} $O(\sqrt{T})$ as SSP from {\cite{ho2013more}}. That is, SGD converges in expectation when the number of iterations $T$ is large under DSSP. We first present the theorem of SSP. Based on the theorem, we show that DSSP has a bound on regret. Therefore, DSSP supports SGD convergence following the same conditions as SSP. 
%\null\hfil
\begin{theorem}[adapted from {\cite{ho2013more}}. \textbf{SGD under SSP}]
\label{ssp}
\hfil \break
Suppose function $f(w) := \sum_{t=1}^{T} f_t(w)$ is a convex function and $\forall f_t(w)$ is also convex.
We use iterative convergent optimization algorithm (gradient descent) on one component $\nabla f_t$ at a time to search for the minimizer $w^*$ under SSP with the staleness threshold $s$ and $P$ workers. 
Let $v_t := -\eta_t \nabla f_t(\widetilde{w}_{t})$ where $\eta_t = \frac{\sigma}{\sqrt{t}}$ and $\sigma = \frac{F}{L\sqrt{2(s+1)P}}$. Here $\widetilde{w}_{t}$ represents the noisy state of the globally shared weight. $F$ and $L$ are constants. 
Assume $f_t$ are $L$-Lipschitz with constant $L$ and the distance $D(w\|w')$ between two multidimensional points $w$ and $w'$ is bounded such that $D(w\|w') := \frac{1}{2}\|w-w'\|^2_2 \le F^2$ where $F$ is constant. 
We have a bound on the regret
\begin{equation}{\label{regret_ssp}}
R[X] := \sum_{t=1}^{T} f_t(\widetilde{w}_{t}) - f(w^*) \le 4FL\sqrt{2(s+1)PT}
\end{equation}
Thus, $R[X] = O(\sqrt{T})$ since $\lim_{T\to\infty} \frac{R[X]}{T}=0$
\end{theorem}
%\hfil\null

\begin{theorem}[\textbf{SGD under DSSP}] \label{dssp} 
\hfil Following all \\conditions and assumptions from Theorem \ref{ssp}, we add a new term $R=[0, s_U - s_L]$, the \textit{range} of the staleness threshold.
Let $r \in R$ and $r \ge \forall r' \in R$. 
We have a bound on the regret 
\begin{equation}
R[X] := \sum_{t=1}^{T} f_t(\widetilde{w}_{t}) - f(w^*) \le 4FL\sqrt{2(s_L+r+1)PT}
\end{equation}
Thus, $R[X] = O(\sqrt{T})$ since $\lim_{T\to\infty} \frac{R[X]}{T}=0$
\end{theorem}

\begin{proof}
Since we follow all conditions and assumptions from Theorem \ref{ssp}, we need to show the newly added range $R$ does not change the regret bound of SSP. In DSSP, the threshold is dynamically changing between $s_L$ and $s_U$ where $r \in R$ and $R = [0, s_U - s_L]$. We know that SSP with threshold $s_L$ has a bound on regret according to (\ref{regret_ssp}). We only extend the threshold $s_L$ of Theorem \ref{ssp} to $s_L + r$ where $r$ is the largest number from $R$. Suppose we set a fixed threshold $s'$ for SSP, then our DSSP can be deducted to SSP when we set $s' = s_L + r$. Thus, we have a upper bound on regret of SSP with threshold $s'$.
\end{proof}

\section{Experiments}
We evaluate the performance of DSSP compared to other three  distributed paradigms. We aim to find whether DSSP converges faster than SSP on average and whether it can maintain the predictive accuracy of SSP. 

\subsection{Experiment setup}
\subsubsection{Hardware}
We conducted experiments on the SOSCIP GPU cluster{\cite{soscip}} with up to four IBM POWER8 servers running Ubuntu 16.04. Each server has four NVIDIA P100 GPUs. Each server has 512 GB ram and 2$\times$10 cores. The servers are connected with Infiniband EDR. Each server connects directly to a switch with dedicated 100 Gbps bandwidth. 

We also set up a virtual cluster with a mixed GPU models by creating two Docker containers running Ubuntu 16.04 on a server with NVIDIA GTX1060 and GTX1080 Ti. The server has 64 GB ram and 8 cores. Each container is assigned with a dedicated GPU.
\subsubsection{Dataset}
We used CIFAR-10 and CIFAR-100 datasets{\cite{krizhevsky2009learning}} %ImageNet ILSVRC 2012 dataset{\cite{ILSVRC15}}
for image classification tasks. Both datatsets have 50,000 training images and 10,000 test images. CIFAR-10 has 10 classes, while CIFAR-100 has 100 classes. %The dataset has 1281167 training images and 50000 validation images.
\iffalse
\begin{center}
\begin{tabular}{ c|c|c|c }
%\hline
Dataset & Training images & Test images & Class\\
\hline
CIFAR-10 & 50000 & 10000 & 10\\
%\hline
CIFAR-100 & 50000 & 10000 & 100\\
%\hline
\end{tabular}
\end{center}
\fi
\subsubsection{Models}
We used a downsized AlexNet{\cite{krizhevsky2012imagenet}}, 
ResNet-50 and ResNet-110{\cite{he2016deep}} as our deep neural network structure to evaluate the four distributed paradigms. We reduced the original AlexNet structure to a network with 3 convolutional layers and 2 fully connected layers
%3-layer convolutional and 2-layer fully connected network 
to achieve faster convergence within 24 hours (which is the time limit we are allowed to run for each job on the SOSCIP cluster). We set the staleness threshold $s_L=3$ and the range $R = [0,12]$ for DSSP which is equivalent to the corresponding threshold range $[3,15]$ for SSP.

%\begin{tabular}{ c|c|c }
%\hline
%Paradigm & Staleness threshold ($s$) & Relaxed range ($R$)\\
%\hline
%SSP & 3 - 15 & $-$ \\
%\hline
%DSSP & 3 & 12 \\
%\hline
%\end{tabular}

\subsubsection{Implementation}
We ran each paradigm on 4 servers. Each server represents a worker which has 4 GPUs. Each GPU loads a copy of the DNN model. Thus, there are 16 replica models for 4 workers. Each worker collects the computed gradients from 4 GPUs by the end of every iteration and sends the sum of the gradients to the parameter server. One of the 4 servers is also elected to run the parameter server when the training starts from the very beginning as defined in MXNet. %Each experiment can only run 24 hours due to the maximum hour per job allowance by the SOSCIP cluster. 
We ran each experiment three times and chose the medium result based on the test accuracy. %The purpose of the experiments is to observe the performance of four distributed paradigm. Hence, we did not fine-tuning the selected DNN models.

\subsection{Results and discussion}
We used batch size 128, learning rate 0.001 in 300 epochs to train the downsized AlexNet on CIFAR-10. Figure \ref{fig:alex_test} shows that DSSP, SSP and ASP converge much faster than BSP, and that DSSP and SSP converge to a higher accuracy than ASP. BSP is the slowest to complete the 300 epochs. The performance of DSSP and averaged SSP are similar, with DSSP converging a little bit faster to a bit higher accuracy. 
DSSP and averaged SSP complete 300 epochs almost at the same time. 
Note that this result is expected because the result of averaged SSP is the average over the results from 13 different threshold values from 3 to 15, and when its threshold is large, it is very fast, much faster than DSSP with a threshold range of [3,15]. However, a larger threshold of SSP incurs more staler gradients, which implies more noises and decreases the quality of iterations{\cite{ho2013more}}. Theoretically, as the threshold $s$ of SSP increases, the rate of convergence decreases per iteration update{\cite{ho2013more}}.
Figure \ref{fig:alex_ssp} compares DSSP with individual SSPs with different threshold values. It shows that DSSP converges a bit faster to a higher accuracy than almost all of the SSPs except for one.  
%but has lower accuracy than a specific threshold SSP ($s$=12) in Figure \ref{fig:alex_ssp}. 
%ASP is the fastest while BSP is the slowest to complete the training. They are also the lower bound and upper bound of the test accuracy among all. The convergence rate from high to low are DSSP, ASP, SSP and BSP in Figure \ref{fig:alex_test} because our DSSP delivers more weight updates than others within a time unit meanwhile it produces less staled gradients than ASP and SSP.

For ResNet-50 and ResNet-110 training on CIFAR-100, we used batch size 128, learning rate 0.05 and decay 0.1 twice at epoch 200 and 250 in 300 epochs for both. In Figure \ref{fig:resnet50_test}, DSSP has the same convergence rate as ASP and SSP, and they converges much faster than BSP although BSP completes 300 epochs faster than others on both ResNets. 
%However, BSP converges the slowest among all since under the same hyper-parameters setting and the iterative convergent method (e.g. SGD), the other three paradigms deliver more updates (larger iteration throughput) than BSP to the model parameters. 
Again, DSSP converges a little faster and archives a bit higher accuracy than averaged SSP in Figure {\ref{fig:Resnet110_test}}.%By now, we observe that four distributed paradigms have similar trends on AlexNet and Vgg16: BSP, SSP, DSSP and ASP are in descending order in terms of the training time cost for the same number of epochs. SSP achieves higher accuracy than DSSP and ASP. 

%To train ResNet152, we use learning rate 0.1 and decay 0.1 three times at epoch 30, 60 and 90. The batch size is set to 128. DSSP converges faster than ASP but slower than SSP and achieves higher accuracy than SSP and BSP (see Figure \ref{fig:res2nodes}). As the number of nodes increases, DSSP converges faster than SSP and receives the highest accuracy (see Figure \ref{fig:res4nodes}). However, this model is different from the above two. It does not contain fully connected layers. The speed or iteration throughput trend in ascending order becomes ASP, DSSP, SSP and BSP. In fact, ASP and DSSP are very close on completing the same given epochs except DSSP does not diverge as ASP (in Figure \ref{fig:res4nodes}). Oppositely, BSP has the similar performance as ASP in aforementioned two DNNs. It obtains the lowest accuracy but uses the least training time to complete the same epochs among all paradigms.   

%Inception-v3
%To train Inception-v3, we follow {\cite{Szegedy_2016_CVPR}} which uses learning rate 0.045 and decay 0.94 every 2 epochs. We set the batch size to 256. The results (in Figure \ref{fig:incept2nodes},  \ref{fig:incept4nodes}) show that DSSP converges faster than SSP on 4 nodes but has performance close to SSP on 2 nodes. For Inception-v3, the performance variance on four distributed paradigms is small which is contrarily to the three large-size DNNs (in Table \ref{table_epochs}). %Note that its model size is 1/6 of Vgg16 and 1/3 of AlexNet and ResNet152 (as shown in Table \ref{table_epochs}). 
Four distributed paradigms on both ResNets behave in an opposite way compared to the downsized AlexNet which has fully connected layers in terms of the time taken to complete 300 epochs. The order from fastest to slowest is BSP, SSP, DSSP and ASP.
%removed with Table with average epoch traing time
%Table \ref{table_epochs} shows that the iteration throughputs of ASP, DSSP, SSP and BSP are in decreasing order for the model with fully connected layers (downsized AlexNet). Less training time per epoch indicates larger iteration throughput. This trend seems to be reversed on DNNs without fully connected layers (ResNet-50, ResNet-100). 

%Experiments show that the DNNs with fully connected layers (except the last fully connected layer functioning as flattening the tensors for softmax input) follows DSSP converges faster than SSP and obtains higher accuracy than ASP.
Based on the empirical results, we observe two opposite trends of ASP, DSSP, SSP and BSP with respect to their performance. The trends can be classified by the architecture of DNNs: ones that contain fully connected layers and ones that do not. Note that we do not count the final fully connected softmax layer as the fully connected layer over the discussion.

\subsubsection{DNNs with fully connected layers (AlexNet)} DSSP converges to a higher test accuracy faster than ASP, BSP and average SSPs in its corresponding staleness threshold range. ASP has the largest iteration throughput and its convergence rate is close to DSSP but it usually converges to a very low accuracy (the lowest of four paradigms) and diverges sometimes (see Figure \ref{fig:alex_test}). DSSP performs between SSP and BSP in terms of final test accuracy. In this category, our DSSP converges faster than SSP and ASP to a higher accuracy. %BSP completes the 300 epochs the last whereas other three paradigms use less time and completes around the same time with small variance. 
We know BSP guarantees the convergence and its accuracy is the same as using a single machine. Thus, it has no consistency errors caused by delayed updates. Given abundant time of training, BSP can reach the highest accuracy among all distributed paradigms. We do not discuss BSP in detail here since our focus is to show the benefits that our DSSP brings compared to SSP and ASP, both cost less training time than BSP.

\subsubsection{DNNs without fully connected layers (ResNet-50, ResNet-110)} DSSP converges faster than average SSP in its corresponding range on very deep neural networks. ASP appears to be a strong rival but it has no guarantee to converge as addressed in section {\ref{asp_converage}}. BSP delivers the highest iteration throughput. However, it converges slower and to a lower accuracy than other three paradigms mostly (see Figure \ref{fig:resnet50_test}, \ref{fig:Resnet110_test}). 
The iteration throughputs 
%The training time 
of ASP, DSSP, SSP and BSP are in ascending order. %descending order roughly despite there is only slight difference among ASP, DSSP and SSP.
In this category, the convergence rate of DSSP, ASP and SSP are very close. DSSP performs slightly above the average SSPs where the threshold $s$ starts from 3 to 15 (see Figure \ref{fig:Resnet50_ssp}, \ref{fig:resnet110_ssp}).

\begin{figure*}[ht!]
%\begin{subfigure}{.5\textwidth}
%  \centering
%  \includegraphics[width=0.85\linewidth]{alexnet2nodes71epoch_v2.pdf}
%  \caption{AlexNet trained 71 epochs on 2 nodes}
%  \label{fig:sfig1}
%\end{subfigure}%
%\begin{subfigure}{.5\textwidth}
%  \centering
%  \includegraphics[width=0.85\linewidth]{alexnet2nodes71epoch_v2.pdf}
%  \caption{AlexNet trained 71 epochs on 4 nodes}
%  \label{fig:sfig2}
%\end{subfigure}
%\medskip
\iffalse
\begin{subfigure}{.45\textwidth}
  \centering
  \includegraphics[width=.93\linewidth]{experiments/downsizedAlexNet_ssp_train_loss.pdf}
  \caption{downsized AlexNet trained 300 epochs on CIFAR-10}
  \label{fig:alex_ssp_loss}
\end{subfigure}\hspace*{\fill}
\begin{subfigure}{.45\textwidth}
  \centering
  \includegraphics[width=.93\linewidth]{experiments/downsizedAlexNet_train_loss.pdf}
  \caption{downsized AlexNet trained 300 epochs on CIFAR-10}
  \label{fig:alex_test_loss}
\end{subfigure}
\medskip
\fi

\begin{subfigure}{.45\textwidth}
  \centering
  \includegraphics[width=.93\linewidth]{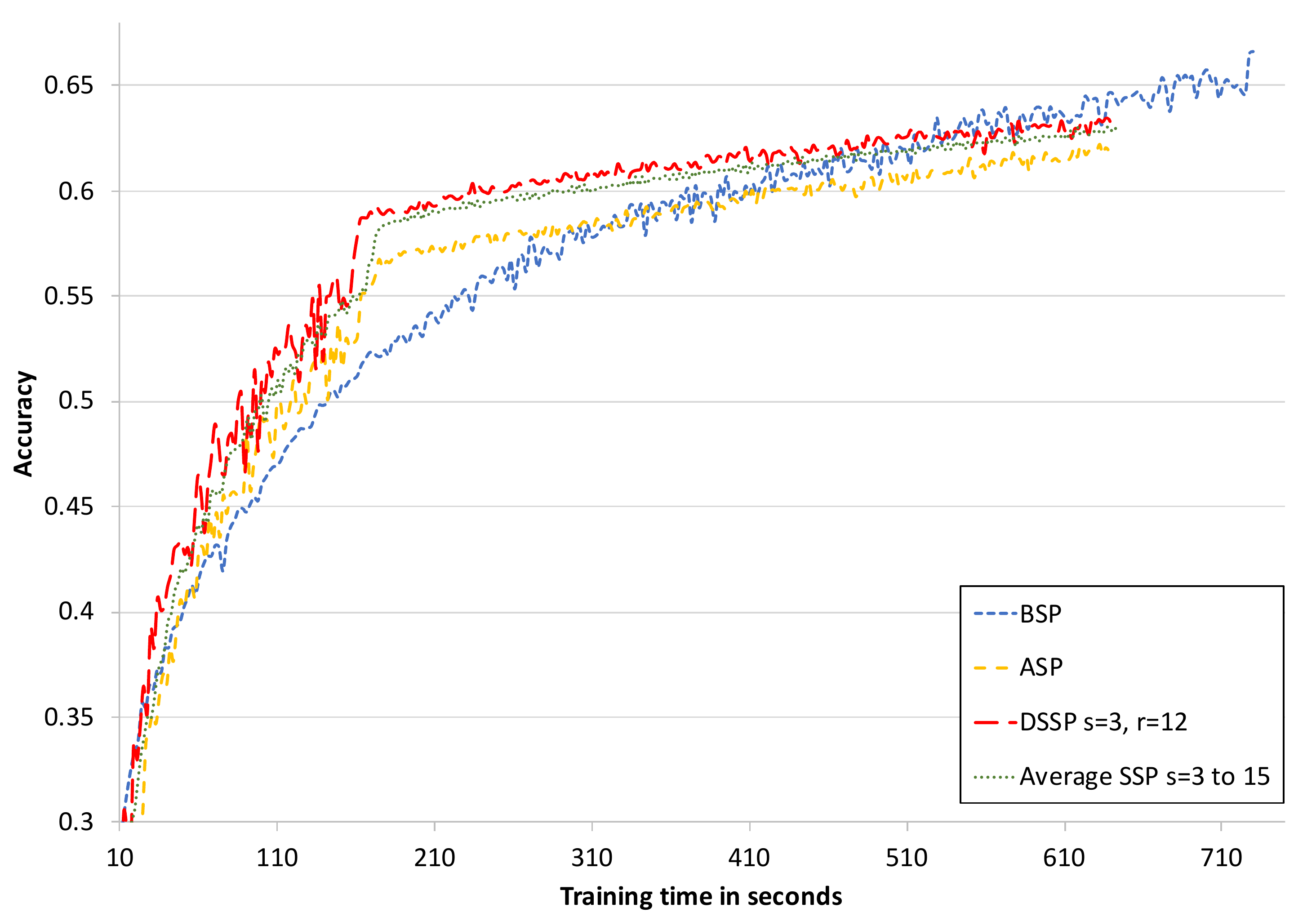}
  \caption{All paradigms run on downsized AlexNet}
  \label{fig:alex_test}
\end{subfigure}
\begin{subfigure}{.45\textwidth}
  \centering
  \includegraphics[width=.93\linewidth]{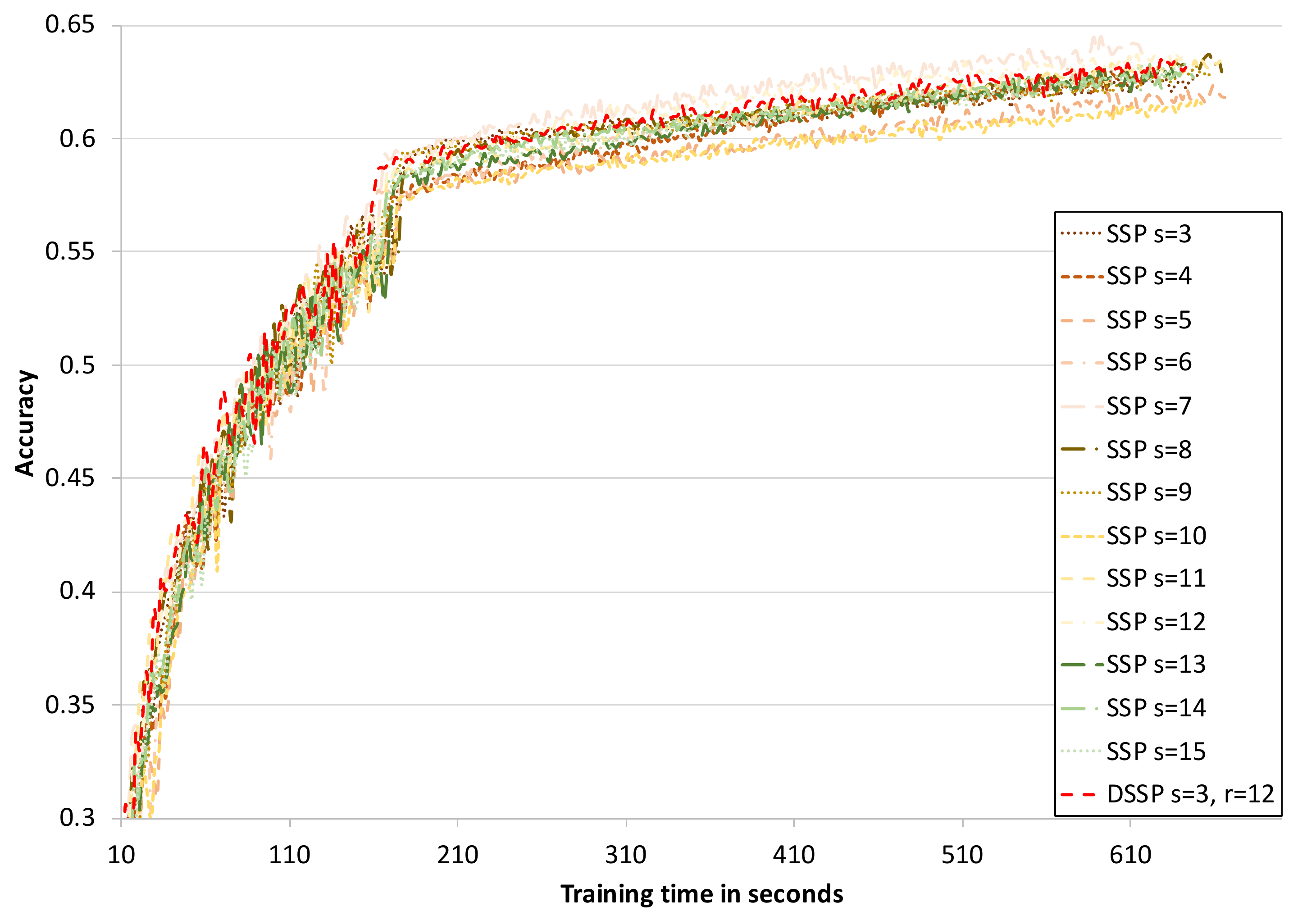}
  \caption{DSSP and SSPs run on downsized AlexNet}
  \label{fig:alex_ssp}
\end{subfigure}\hspace*{\fill}
\medskip
\vspace{0.005in}

\begin{subfigure}{.45\textwidth}
  \centering
  \includegraphics[width=.93\linewidth]{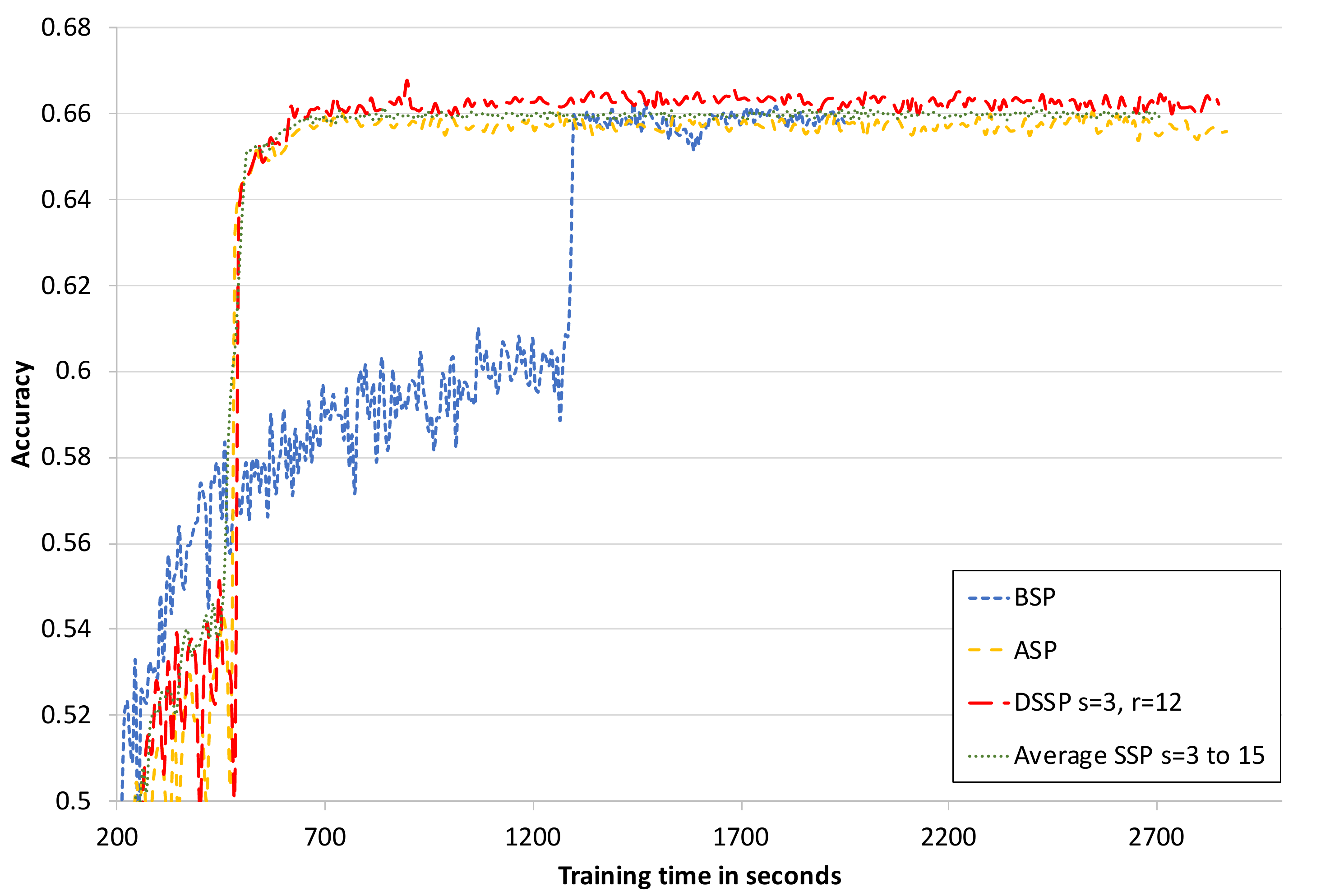}
  \caption{All paradigms run on ResNet-50}
  \label{fig:resnet50_test}
\end{subfigure}
\begin{subfigure}{.45\textwidth}
  \centering
  \includegraphics[width=.93\linewidth]{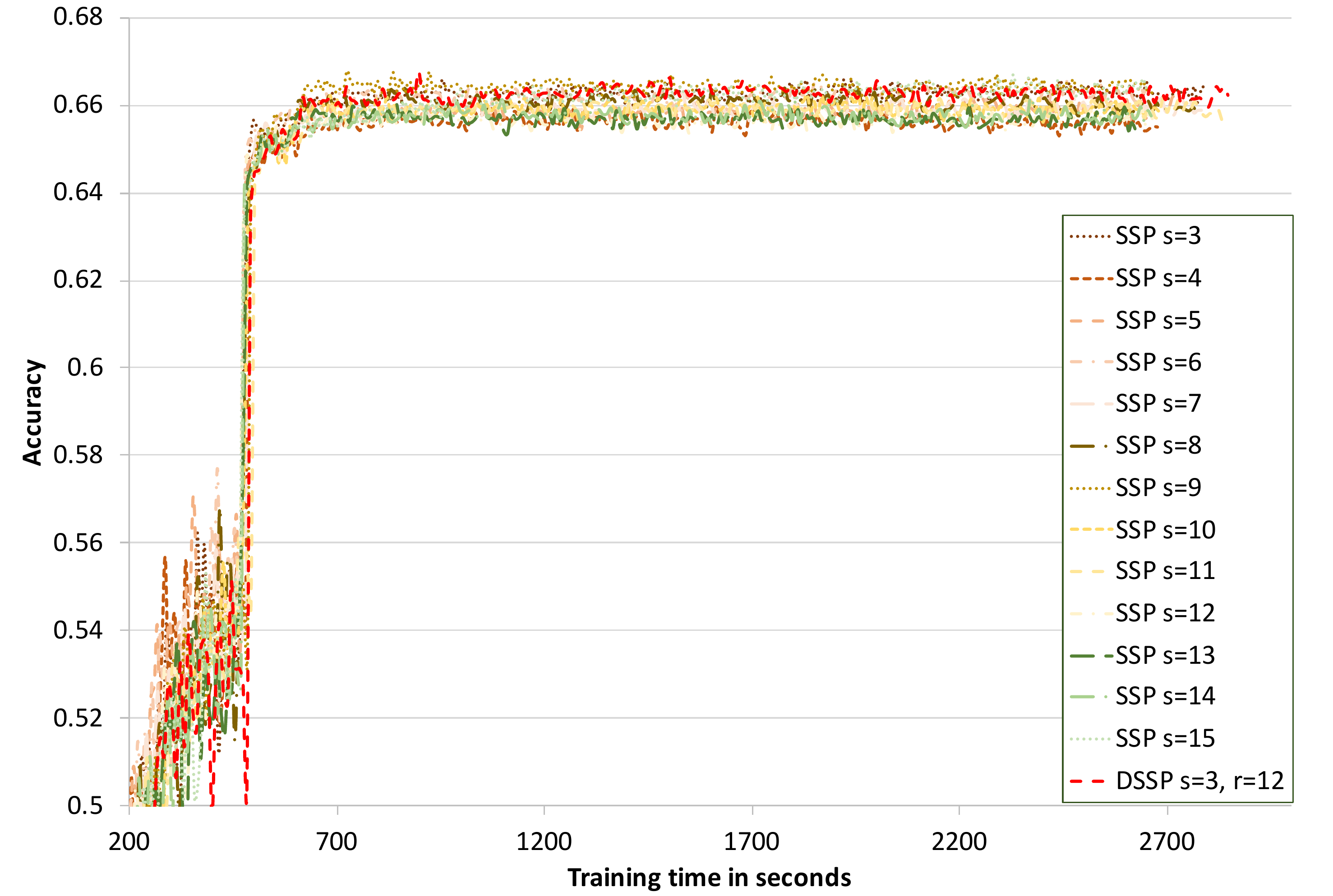}
  \caption{DSSP and SSPs run on ResNet-50}
  \label{fig:Resnet50_ssp}
\end{subfigure}\hspace*{\fill}
\medskip
\vspace{0.005in}

\begin{subfigure}{.45\textwidth}
  \centering
  \includegraphics[width=.93\linewidth]{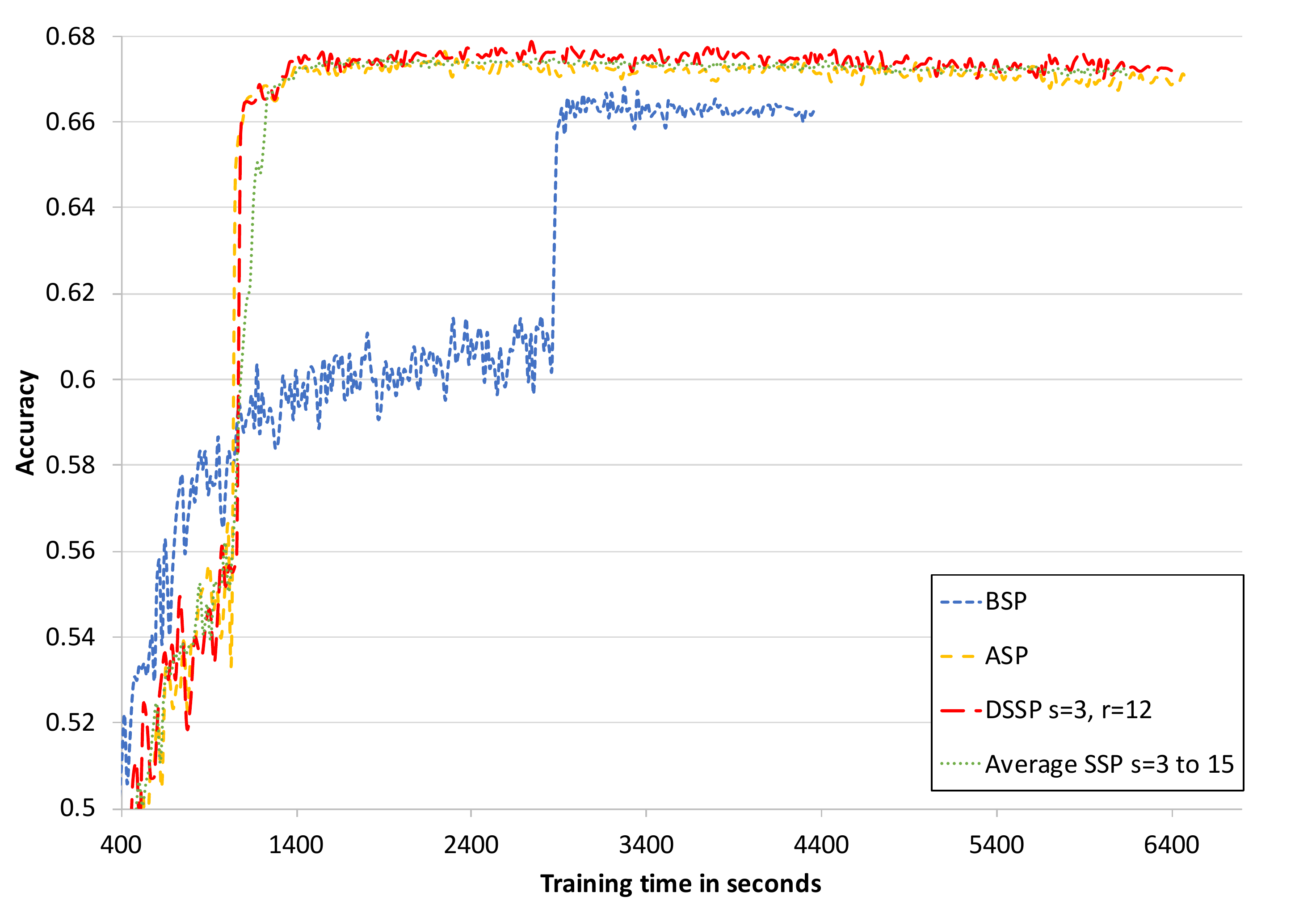}
  \caption{All paradigms run on ResNet-110}
  \label{fig:Resnet110_test}
\end{subfigure}
\begin{subfigure}{.45\textwidth}
  \centering
  \includegraphics[width=.93\linewidth]{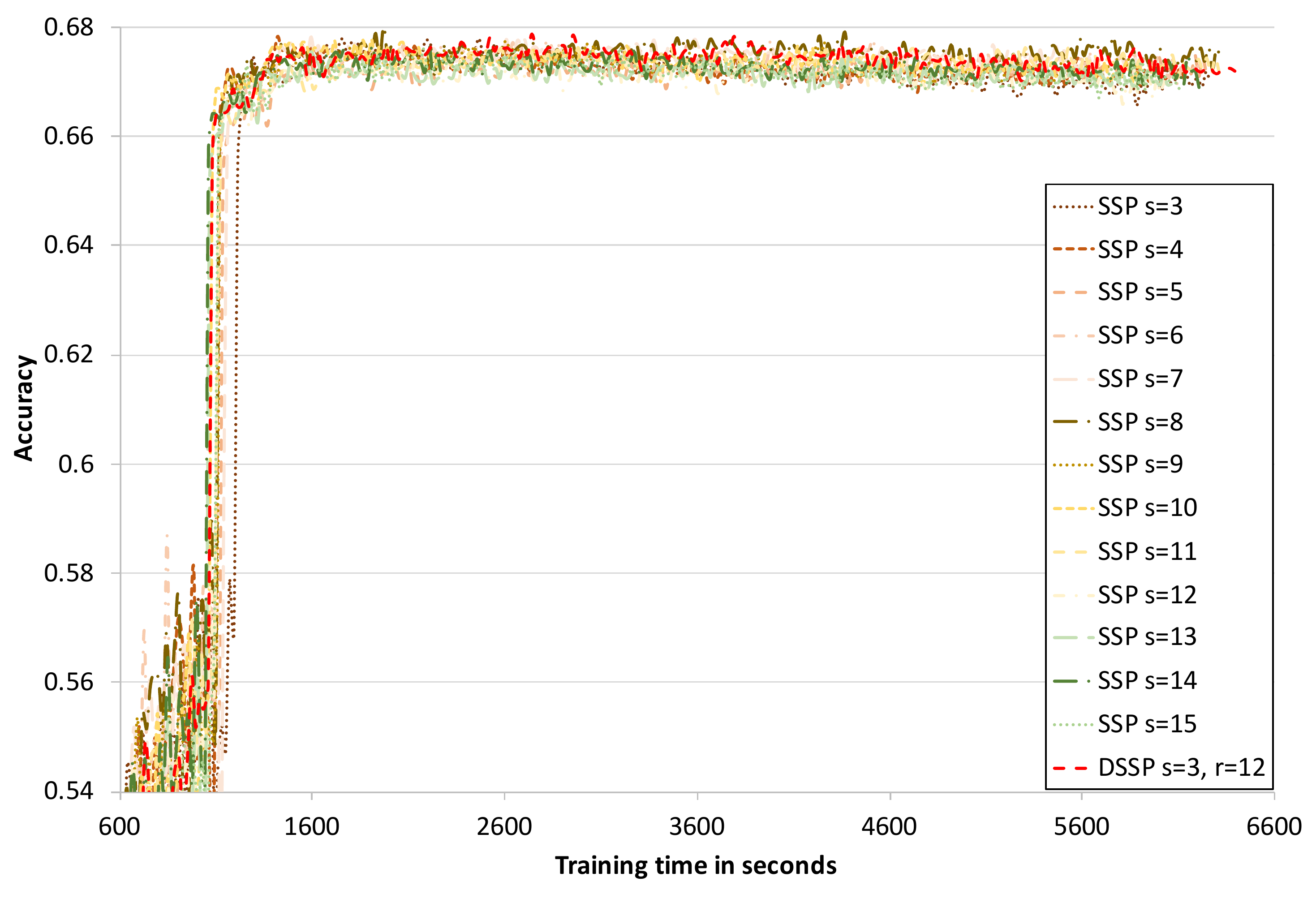}
  \caption{DSSP and SSPs run on ResNet-110}
  \label{fig:resnet110_ssp}
\end{subfigure}\hspace*{\fill}
\caption{Distributed paradigms comparison on downsized AlexNet, ResNet-50 and ResNet-110 training for 300 epochs. Downsized AlexNet is trained on CIFAR-10 and both ResNets are trained on CIFAR-100. Average SSP on the right column is derived by averaging SSPs with threshold from 3 to 15 on the left column. Faster convergence to a targeted high accuracy indicates less training time is required for the paradigm.}%DSSP is the second fastest on completing the given epochs. ASP is the fastest but unstable since it diverges sometime.}
\label{fig:results}
\end{figure*}

\subsection{Demystify the difference}
% from email
Below we answer two questions:
%\begin{enumerate}
 %\item 
 (1) Why does the iteration throughput have the opposite trends for ASP, DSSP, SSP and BSP on DNNs with and without fully connected layers?
 %\item 
 (2) Why do pure convolutional neural networks (CNNs) receive a higher accuracy from DSSP, SSP and ASP than BSP?
%\end{enumerate}
We temporarily name DNNs with fully connected layers as DNNs and pure CNNs as CNNs for the convenience of discussion.

To answer the first question, we observe the difference between the two types of DNNs (with or without fully connected layers): \ding{172} A fully connected layer requires more parameters than a convolutional layer which uses shared parameters{\cite{tompson2015efficient}}. DNNs with fully connected layers have a large number of model parameters that need to be transmitted between workers and the server for updates. \ding{173} Convolutional layers require intensive computing time for matrix dot product operations while computing for fully connected layers involves simple linear algebra operations {\cite{krizhevsky2014one}}. CNNs that only use convolutional layers take a lot of computing time, while their relatively smaller-size model parameters cost less data transmission time between workers and the server than DNNs. Moreover, when the ratio of computing time and communication time per iteration is small, less time can be saved per iteration for workers since computing time per iteration (or one mini-batch) is fixed for a model per worker. To the contrary, when the ratio is large, the communication time per iteration for each worker can be shifted by asynchronous-like parallel schemes and more time can be saved. % in a ``communicating while computing among workers"{\cite{xing2015petuum}} fashion.
Therefore, DSSP, SSP and ASP take less training time on DNNs whereas BSP costs the least training time on CNNs. %The control experiment(\ref) demonstrates our observation analysis. With the same amount of parameters {\cite{desmisify dist}} and the batch size, AlexNet costs more time than ResNet152 for a epoch. ResNet152 processes less images per second than AlexNet. 

%check after computing the weight size
%An obvious evidence can be found from Table \ref{table_epochs} where downsized AlexNet has less a model size than ResNet-50 but have a large gap between averaged epoch time. AlexNet has 3 convolutional layers and 1 fully connected layers. ResNet-50 contains 50 convolutional layers. 

To the second question, the answer lies in the difference between fully connected layers and convolutional layers.
Fully connected layers are easy to overfit the data set due to its large number of parameters{\cite{srivastava2014dropout}}. Thus, any error introduced by staled updates can cause many parameters diverge in non-uniform convergence {\cite{xing2016strategies}}. Informally, fully connected layers overfit to the errors injected by delayed updates or noise.
Convolutional layers have less parameters due to the use of filters (shared parameters). For image classification tasks, a commonly used trick to train CNNs on a small data set is to increase the data by distorting the existing images and saving them {\cite{perez2017effectiveness}} since CNNs are able to tolerate certain scale variations{\cite{xu2014scale}}. Distortion can be done by rotating the image, setting one or two of RGB pixels to zero or adding Gaussian noise to the image {\cite{kang2015simultaneous}}. It is basically to introduce noise to images so that CNN models receive enhanced training and the predictions are improved. 
The errors caused by (not too) staled updates can give the same effect to the training model as the distortion. Figures \ref{fig:resnet50_test}, \ref{fig:Resnet110_test} are good evidence to support that. Furthermore,  {\cite{neelakantan2015adding}} empirically shows that adding gradient noises improves the accuracy for training very deep neural networks which also happened in our ResNet-110 experiments (in Figure \ref{fig:Resnet110_test}).

\subsection{Cluster with mixed GPU models}
The results of DSSP and SSPs on ResNet-110 (see Figure{\ref{fig:Resnet110_test}}) do not show a significant difference on convergence rate on a homogeneous environment where GPUs are identical. Nonetheless, on the heterogeneous environment where we have one GTX1060 and one GTX1080 Ti running on each worker, DSSP converges faster and to a higher accuracy than SSP. We repeated the exact same experiments on ResNet-110 as earlier: use the same hyperparameters setting, run 3 trials on each paradigm and choose the medium one based on the test accuracy. Figure {\ref{fig:resnet110_mixed_gpus}} and Table {\ref{table_mixed_gpus}} show that DSSP can reach a higher accuracy significantly faster than SSP. 
%Therefore, its training time can be  reduced more than half by terminating its training earlier (close to 3050 seconds training) than SSP. Note that the fastest SSP in Table {\ref{table_mixed_gpus}} converges to the same accuracy as DSSP after 6909 seconds training. 
The heterogeneous environment is very common in industry since new GPU models come to the market every year while the old models are still in use. ASP can fully utilize the individual GPU and achieves the largest iteration throughout. However, ASP also introduces the most staled updates among all distributed paradigms. It may converge to a lower accuracy than DSSP when the GPU models' processing capacities are significantly different since the iterations between the fastest worker and the slowest worker are dramatically different. In contrast, DSSP has consistent performance regardless the running environment since it adapts to the environment by adjusting the threshold dynamically for every iteration of workers.

\begin{figure}[!ht]
    \centering
    \includegraphics[width=0.483\textwidth]{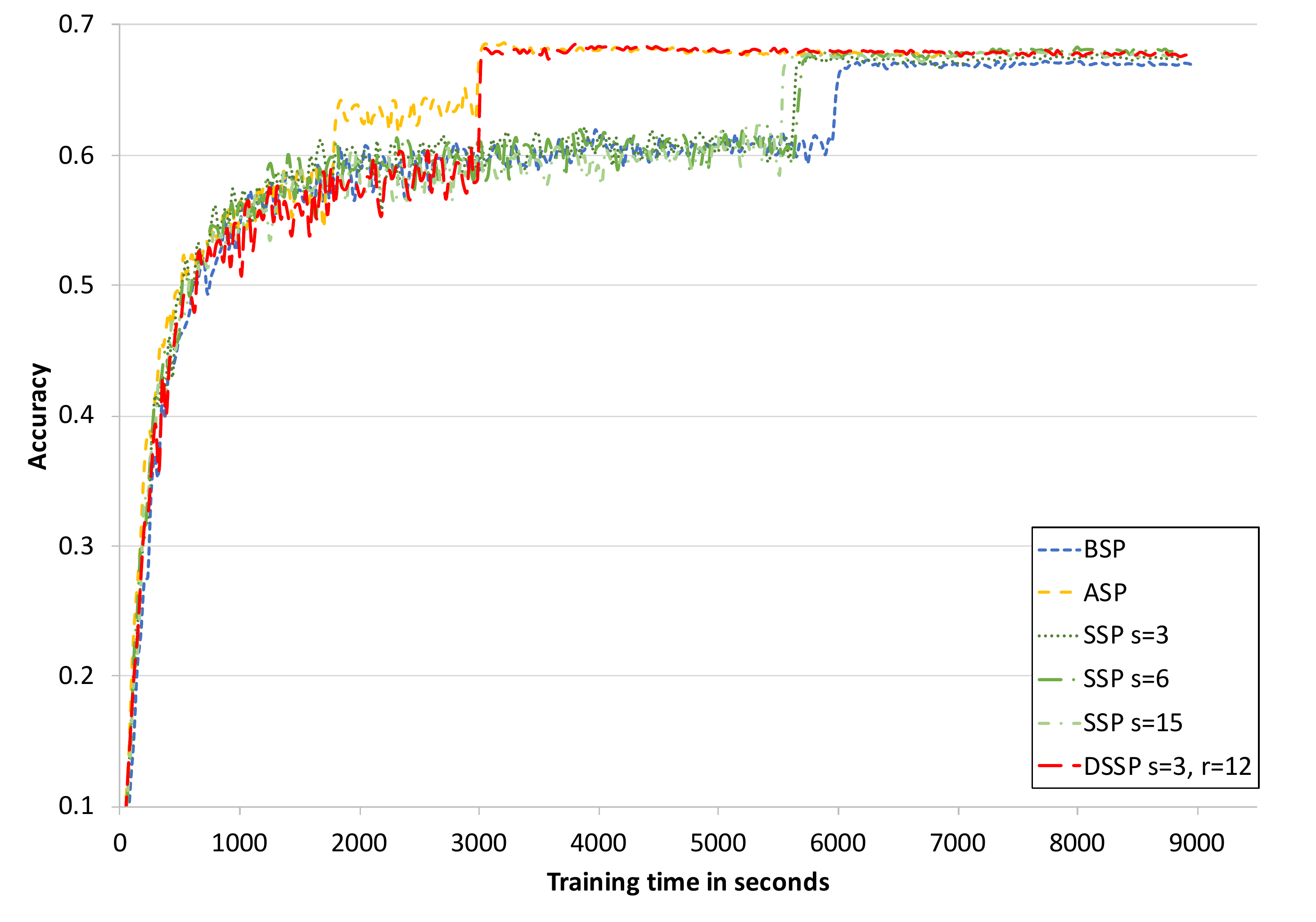}
    \caption{Trained ResNet-110 on CIFAR-100 with two workers on a mixed GPU cluster for 300 epochs. GTX1060 and GTX1080 Ti are assigned to individual worker. Our DSSP converges faster and achieves higher accuracy than SSP.}
    \label{fig:resnet110_mixed_gpus}
\end{figure}

\begin{table}[!t]
\renewcommand{\arraystretch}{1.5}
\centering
\begin{tabular}{ c|c|c }
Distributed & Time to reach & Time to reach \\ 
Paradigm & 0.67 accuracy & 0.68 accuracy\\
\hline
BSP & 6159.2 & $-$ \\ 
ASP & 2993.1 & 3017.2 \\
SSP $s$=3 & 5678.2 & $-$ \\
SSP $s$=6 & 5703.8 & 6908.2 \\
SSP $s$=15 & 5564.9 & 7255.6 \\
DSSP $s_L$=3, $r$=12 & 3016.4 & 3046.3 \\
\end{tabular}
\caption{Time in seconds to reach the targeted test accuracy in training. The maximum test accuracy of BSP and SSP with $s$=3 is 0.67. Trained ResNet-110 on CIFAR-100 with two workers for 300 epochs. Each worker has either GTX1080 Ti or GTX1060.}
\label{table_mixed_gpus}
\end{table}

\section{Conclusion}
We presented dynamic staleness synchronous parallel (DSSP) paradigm for distributed training using the parameter server framework. DSSP improves SSP in the sense that with DSSP a user does not need to provide a specific staleness threshold which is hard to determine in practice, and also that DSSP can dynamically determine the value for the threshold from a range using a lightweight method according to the run-time environment. This does not only alleviate the burden of an exact manual staleness threshold selection or multiple trials of hyperparameter selection, but it also provides flexibility of selecting different thresholds for different workers at different times. We provided theoretical analysis on the expected convergence of DSSP which inherits the same regret bound of SSP to show that DSSP converges in theory as long as the range is constant. We evaluated DSSP by training three DNNs on two datasets and compared its results with other distributed paradigms.
For DNNs without fully connected layers, DSSP achieves higher accuracy than BSP and slightly better accuracy than averaged SSP. For DNNs with fully connected layers, DSSP generally converges faster than BSP, ASP and averaged SSP to a higher accuracy even though BSP can eventually reach the highest accuracy if it is given more training time. Unlike ASP, DSSP ensures the convergence of DNNs by limiting the staled delays. DSSP gives significant improvement than SSP and BSP in a heterogeneous environment with mixed models of GPUs, converging much faster to a higher accuracy.
DSSP also shows more stable performance on either homogeneous or heterogeneous environment compared to other three distributed paradigms. 
Furthermore, we discussed the difference in the trends of four distributed paradigms on DNNs with and without fully connected layers and the potential causes. For the further work, we will investigate how DSSP can adapt to an unstable environment where network connections are fluctuating between the servers.

\section*{Acknowledgment}

This work is funded by the Natural Sciences and Engineering Research Council of Canada (NSERC), IBM Canada and the Big Data Research Analytics and Information Network (BRAIN) Alliance established by Ontario Research Fund - Research Excellence Program (ORF-RE).
The experiments reported in this paper were performed on the GPU cluster of SOSCIP. SOSCIP is funded by the Federal Economic Development Agency of Southern Ontario, the Province of Ontario, IBM Canada, Ontario Centres of Excellence, Mitacs and 15 Ontario academic member institutions.

% trigger a \newpage just before the given reference
% number - used to balance the columns on the last page
% adjust value as needed - may need to be readjusted if
% the document is modified later
%\IEEEtriggeratref{8}
% The "triggered" command can be changed if desired:
%\IEEEtriggercmd{\enlargethispage{-5in}}

% references section

% can use a bibliography generated by BibTeX as a .bbl file
% BibTeX documentation can be easily obtained at:
% http://www.ctan.org/tex-archive/biblio/bibtex/contrib/doc/
% The IEEEtran BibTeX style support page is at:
% http://www.michaelshell.org/tex/ieeetran/bibtex/
%\bibliographystyle{IEEEtran}
% argument is your BibTeX string definitions and bibliography database(s)
%\bibliography{IEEEabrv,../bib/paper}
%\newcommand{\BIBdecl}{\setlength{\itemsep}{0.25 em}}
\bibliographystyle{IEEEtran}
\bibliography{bare_conf}
%
% <OR> manually copy in the resultant .bbl file
% set second argument of \begin to the number of references
% (used to reserve space for the reference number labels box)
%\begin{thebibliography}{1}

%\bibitem{IEEEhowto:kopka}
%H.~Kopka and P.~W. Daly, \emph{A Guide to \LaTeX}, 3rd~ed.\hskip 1em plus
%  0.5em minus 0.4em\relax Harlow, England: Addison-Wesley, 1999.

%\end{thebibliography}

% that's all folks
\end{document}